\documentclass{llncs}

\usepackage{amsmath}
\usepackage{amssymb}
\usepackage{amsfonts}
\usepackage{amssymb}

\usepackage{graphicx}

\usepackage{algorithm}
\usepackage{algpseudocode}

\newcommand{\fun}[1]{\mathrm{#1}}

\newcommand{\mket}[1]{| #1 \rangle}

\newcommand{\mbraket}[2]{\langle #1 | #2 \rangle} 
\newcommand{\mtr}[1]{\mathrm{Tr}\left( #1 \right)}
\newcommand{\mptr}[2]{\mathrm{Tr}_{#2}\left( #1 \right)}

\newcommand{\keywords}[1]{\par\addvspace\baselineskip\noindent\keywordname\enspace\ignorespaces#1}

\begin{document}


\mainmatter

\title{Sorting of quantum states with respect to amount of entanglement included}

\author{Roman Gielerak \and Marek Sawerwain}
\authorrunning{Roman Gielerak \and Marek Sawerwain}
\tocauthor{Roman Gielerak and Marek Sawerwain}

\institute{Institute of Control \& Computation Engineering \\
University of Zielona G\'ora, ul. Podg\'orna 50, Zielona G\'ora 65-246, Poland \\
\email{{R.Gielerak,M.Sawerwain}@issi.uz.zgora.pl}
}

\maketitle

\begin{abstract}
The canonical Schmidt decomposition of quantum states is discussed and its implementation to the Quantum Computation Simulator is outlined. In particular, the semiorder relation in the space of quantum states induced by the lexicographic semiorder of the space of the Schmidt coefficients is discussed. The appropriate sorting algorithms on the corresponding POSETs consisting from quantum states are formulated and theirs computer implementations are being tested.
\end{abstract}

\section{Introduction}

Let $\mathcal{E}(H_A \otimes H_B)$ be a set of quantum states of a given composite system $S_{A+B}$ on the corresponding Hilbert space $H_A \otimes H_B$. As it is well known the problem to decide whether a given state $\rho \in \mathcal{E}$ is entangled or not is in general NP-HARD \cite{4}. Although some mathematical procedures for this purpose and also for deriving in a quantitative way the amount of entanglement included do exist the real problem with them is that they are hardly to be efficiently calculable \cite{1} and \cite{2}.

In the case of pure states the operational semiorder $\mket{\psi} \stackrel{LOCC}{\succ} \mket{\varphi}$ on the space of pure states (meaning that the state $\mket{\psi}$ can be transformed into the vector $\mket{\varphi}$ by exclusive use of only LOCC class of operations) has been formulated in an effectively calculable way by appealing to the corresponding Schmidt decomposition in Nielsen \cite{3}. The relation $\stackrel{LOCC}{\succ}$ introduces only partial order on the space of pure states and it is why the adaptations of standard sorting algorithms of the corresponding POSETs are much more involved \cite{5} and \cite{6} comparing to the case of linearly ordered sets. Several adaptations of this sorting procedures are being adopted and tested on the Zielona G\'ora Quantum Computing Simulator and some recent results of this kind will be presented in the present contribution, see also [7].

Another topic discussed in our contribution is an attempt to generalize the Nielsen result to the case of general mixed quantum states. For this purpose the Schmidt decomposition in the corresponding Hilbert-Schmidt space has been used and certain semiorder relation in the space of quantum states has been introduced. The effort has been made to connect the introduced semiorder with several notions of quantitative measures of entanglement.

A well known distillation of entanglement procedure \cite{Chuang_QCQI} also introduces a partial semiorder on the space of quantum states. However this process is hardly to be effectively calculable and moreover it requires to have many (infinitely many in fact) copies of a given unknown quantum state at hands in order to perform the distillation process.

Although we have no complete proof we formulate a conjecture that the partial order induced by the lexicographic order of the Schmidt decomposition coefficients is connected to the operational meaning saying that in the state $\rho_1$ is no less entanglement contained then in the state $\rho_2$ if the state $\rho_1$ may be transformed into $\rho_2$ by means of local operations supplemented by classical communication only.


\section{Algorithms for sorting quantum states}

\subsection{Canonical Schmidt decomposition of quantum states}

For a given finite-dimensional Hilbert space $\mathcal{H}$ the corresponding Hilbert-Schmidt space is denoted as $HS(\mathcal{H})$. Let us recall that the space $HS(\mathcal{H})$ consist of all linear operations acting on $\mathcal{H}$ and equipped with the following scalar product:
\begin{equation}
\mbraket{A}{B}_{HS} = \mtr{A^{\dagger}B}
\end{equation}
A system $(E_i,i=1,\ldots,\dim(\mathcal{H})^2)$ of linearly independent matrices on $\mathcal{H}$ is called complete orthonormal system iff $\mbraket{E_i}{E_j}_{HS}=1$. If moreover all $E_i$ are hermitean the system $(E_i)$ is called complete hermitean orthonormal system.
\begin{proposition}
Let $\rho \in \mathcal{E}(\mathcal{H}_A \otimes \mathcal{H}_B)$. Then there exist: a number $r > 0$ (called the canonical Schmidt rank of $\rho$) and a complete orthonormal system $(E_i^A)$ (resp. $(E_j^B)$) in $HS(\mathcal{H}_A)$ (resp. $HS(\mathcal{H}_B)$) and such that
\begin{equation}
\rho = \sum_{\alpha=1}^r \lambda_{\alpha} E^A_{\alpha} \otimes E^B_{\alpha}
\label{lbl:eqn:schmidt:decomposition}
\end{equation}
where the numbers $\lambda_{\alpha} > 0$ are called (the canonical) Schmidt coefficients of $\rho$. If all $\lambda_{\alpha}$ are different then this decomposition is unique.
\end{proposition}
\begin{remark}
A different notions of Schmidt decomposition of density matrices are being discussed in the literature \cite{Terhall}. Our Schmidt characteristics like the canonical Schmidt rank and (canonical) Schmidt coefficients and the corresponding orthonormal systems are in unique way connected with a given $\rho$ and in principle all the properties (separability|nonseparability for example) of $\rho$ should be obtainable form this decomposition. For example if the corresponding $E^A_{\alpha}$, $E^B_{\alpha}$ in formula (\ref{lbl:eqn:schmidt:decomposition}) are nonnegative and therefore hermitean then $\rho$ is separable.
\end{remark}

\begin{remark}
It is well known \cite{Chen}, \cite{Rudolph} that for separable states the sum of the canonical Schmidts coefficients is always less or equal to 1. This leads to the separability criterion known as cross norm criterion.
\end{remark}

\begin{remark}
The closed subspace of $HS(\mathcal{H})$  consisting of hermitean matrices forms a real Hilbert space. Therefore if the SVD theorem extends to the real Hilbert space case then the corresponding systems in formula (\ref{lbl:eqn:schmidt:decomposition}) are hermitean by the very construction.
\end{remark}

Now we formulate constructive route to the canonical Schmidt decomposition.
\begin{proposition}
Let $d=\dim(\mathcal{H})$ and let $(E_i,i=1,\ldots,d^2)$ be a system of linear independent matrices on $\mathcal{H}$. Then there exists operation $\mathcal{O}$ converting the system $(E_i)$ into the orthonormal system $(F_i)$. If the system $(E_i)$ consists of hermitean matrices then $\mathcal{O}((E_i))$ is also system formed from hermitean matrices.
\end{proposition}
\begin{proof}
The well known Gram-Schmidt orthonormalisation procedure is used as an example of the converting operation $\mathcal{O}$. $\blacksquare$
\end{proof}
Let now $(F^A_i)$ (resp. $(F^B_j)$) be any orthonormal system in $HS(\mathcal{H}_A)$ (resp.$HS(\mathcal{H}_B)$). Then the system $(F^A_i \otimes F^B_j)$ forms a complete orthonormal system in $HS(\mathcal{H}_A \otimes \mathcal{H}_B)$. Thus taking any $\rho \in \mathcal{E}(\mathcal{H}_A \otimes \mathcal{H}_B)$ we can decompose:
\begin{equation}
	\rho = \sum_{i,j=1} c_{i,j} F^A_i \otimes F^B_j
\end{equation}
where $c_{i,j}=\mtr{\rho F^A_i \otimes F^B_j}$.

Then we apply SVD operation to the matrix $C=(c_{i,j})$ yielding (like in the vector case) all the data for supplying the decomposition (\ref{lbl:eqn:schmidt:decomposition}). In particular the singular values of the matrix $C$ are equal to the squares of the Schmidt numbers from (\ref{lbl:eqn:schmidt:decomposition}).


\subsection{Linear and partial semi-order for entanglement states}

Firstly, we present a simple algorithm to realise sorting a set of quantum states 
by using von Neumann entropy notion. We will call this algorithm a linear sorting by 
entropy algorithm (abbreviated as LSEA). The pseudo-code of LSEA 
is presented in Algorithm~(\ref{lbl:linear:sorting:entropy}).

\begin{algorithm}[!ht]
\caption{Algorithm for sorting entangled quantum states using the von Neumann entropy}
{\scriptsize
\begin{algorithmic}[1]
\Function{LSEA}{ $\Sigma : \{ \rho_1, \rho_2, \ldots, \rho_N \}$ } : $\Sigma^{SORT} : \{ \rho_1, \rho_2, \ldots, \rho_N \}$ 
\For{i=1 \textbf{to} N}
\State $\rho^{A}_{i} = \mptr{\rho_i}{H_A}$
\State $[\sigma_i, V_i] = \mathrm{EigenSystem}(\rho^{A}_{i})$
\State $\mathrm{En}(i) = E(\rho^{A}_{i}) = - \sum_{k} \lambda^i_k \log \lambda^i_k$
\EndFor
\State $\Sigma^{SORT}$ = \Call{ClassicalSort}{$\{ \mathrm{En}(1), \mathrm{En}(2), \ldots, \mathrm{En}(N) \}$}
\State \textbf{return} $\Sigma^{SORT}$
\EndFunction
\end{algorithmic}\label{lbl:linear:sorting:entropy}
}
\end{algorithm}

The second presented algorithm realises sorting of entangled states using the Schmidt decomposition. The pseudo-code 
is presented in the Algorithm~(\ref{lbl:partial:order:sorting}). The input of the Algorithm (\ref{lbl:partial:order:sorting}) is now a list $\mathbb{V}$ of vector states on the space $\mathcal{H}=\mathcal{H}_A \otimes \mathcal{H}_B$. The output is divided into two parts. The first part is the partitioning of $\mathbb{V}$: $V=[V_1,\ldots,V_p]$ where
\begin{align}
V(i) \in \mathbb{V}, \mathrm{SchmidtRank}(V_i) = r_i = const; \; \mathrm{and} \\ \notag
1 \leq r_1 < r_2 < \ldots < r_p \leq \min(\dim \mathcal{H}_A, \dim \mathcal{H}_B), \; U_i V(i) = \mathbb{V}
\end{align}
i.e. the partitioning with respect to increasing Schmidt's ranks. Additionally, we return the complete data of merging of each $V_i$.
\begin{figure}[!ht]
\includegraphics[height=3.0cm]{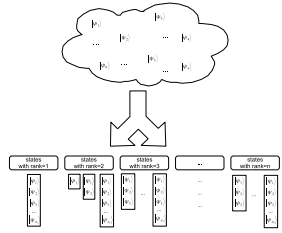}\centering
\caption{General idea of sorting quantum states where the lexicographic order is used. In result the obtained structure represents the partial order where in the sets of quantum states with the same Schmidt rank may appear a linear chains which are non-comparable}
\end{figure}

We need to define the corresponding semiorder relation. For a given pair of quantum states $\rho_1$ and $\rho_2$ we perform the canonical Schmidt decomposition (\ref{lbl:eqn:schmidt:decomposition}) of them obtaining the corresponding Schmidt ranks $r_1$, $r_2$ and the corresponding Schmidt coefficients $\lambda_{\alpha}(1)$, $\lambda_{\alpha}(2)$. Having this data we can formulate the following algorithm which in fact introduces the semiorder relation in the space of quantum states. This algorithm will be called \textsc{SD-QueryOracle}.
\begin{algorithm}[!ht]
\caption{Implementation of query oracle for sorting entangled quantum states}
{
\begin{algorithmic}[1]
\Function{SD-QueryOracle}{ $V : \{ \rho_1, \rho_2 \}$ } : $\{(\rho_1 \prec \rho_2), (\rho_2 \prec \rho_1), \mathrm{non-comparable}\}$
\State \textbf{for}{ i=1 \textbf{to} 2 \textbf{do}} $\{r_i, \{\lambda_{1}(i), \ldots, \lambda_{r_i}(i)\}\} = \mathrm{SchmidtDecomp}(V_i)$
\State \textbf{if} $(r_1 > r_2)$ return $(\rho_1 \prec \rho_2)$
\State \textbf{if} $(r_2 > r_1)$ return $(\rho_2 \prec \rho_1)$
\State sort the sets $\{\lambda_{1}(i)_{i=1,2}\}$ in non-increasing order
\State \textbf{if} $(r_1 = r_2)$ \textbf{if} $(\forall_{j=1,...,r_1}\left( \sum_{i=1}^{j} \lambda^2_i(1) \leq \lambda^2_i(2) \right))$ return $(\rho_1 \prec \rho_2)$
\State \textbf{return} non-comparable
\EndFunction
\end{algorithmic}\label{lbl:oracle:two:vec:locc}
}
\end{algorithm}

Having defined the semiorder $\prec$ we can now formulate one of the possible (see \cite{q12} for other versions) sorting sets of quantum states algorithm that will be called \textsc{MergeSort} type algorithm.

\begin{algorithm}[!ht]
\caption{Algorithm for partial ordered sorting entangled quantum states}
{\scriptsize
\begin{algorithmic}[1]
\Function{ChainMergeSort}{ $\Sigma : \{ \rho_1,\ldots, \rho_N \}$ } : $V$
\State\textbf{for} {i=1 \textbf{to} N \textbf{do}} $\{r(i), \{\lambda_{1}(i), \ldots, \lambda_{r_i}(i)\}\} = S(i) = \mathrm{SchmidtDecomp}(V_i)$
\State $r^{\star} = \max(r(1),\ldots, r(n))$
\For{$\alpha$=1 \textbf{to} $r^{\star}$}
\State $V(\alpha) = [ \; ]$
\State \textbf{for} i=1 \textbf{to} N \textbf{do} \textbf{if} (r(i)=$\alpha$) V($\alpha$)=[V($\alpha$), $\rho_{\alpha}$]
\EndFor
\For{$\alpha$=1 \textbf{to} $r^{\star}$}
	\State choice randomly $\rho \in V(\alpha)$
	\State $R^{\alpha}=\left( \rho, \{ \; \} \right)$
	\State $R^{\alpha}(1)=P^{'}$ ; $R^{\alpha}(1)=\{ \; \}$
	\State $U(\alpha)=V(\alpha) \setminus \{ \rho \}$
	\While{ $U(\alpha) \neq \emptyset $ }
		\State choice $\rho_{i} \in U(\alpha)$
		\State $U(\alpha)=U(\alpha) \setminus \{ \rho \}$
		\State construct a chain decomposition $C(\alpha)=\{ C^{\alpha}_1, C^{\alpha}_2, \ldots, C^{\alpha}_q\}$ of $R^{\alpha}$
	\EndWhile
	\For{i=1 \textbf{to} q}
		\State \textbf{I1:} do binary search on $C^{\alpha}_i$ using 
\Call{SD-QueryOracle}{} to find smallest element (if any) that dominates $\rho_i$
		\State \textbf{I2:} do binary search on $C^{\alpha}_i$ using 
\Call{SD-QueryOracle}{} to find largest element (if any) that dominates $\rho_i$
	\EndFor
	\State infer all results $R^{\alpha}$ of \textbf{I1} and \textbf{I1} into $R^{\alpha}$: 
	\State \hspace{1cm} $R^{\alpha}(1) = R^{\alpha}(1) \cup \rho$
	\State \hspace{1cm} $R^{\alpha}(2) = R^{\alpha}(2) \cup R^{\alpha}$
	\State find a chain decomposition $C^{\alpha}$ of $R^{\alpha}$ 
	\State construct ChainMerge($R^{\alpha}$, $C^{\alpha}$) data
\EndFor
\State return ($V=[V_1,V_2,\ldots,V_r]$, ChainMerge(V)=[ChainMerge($V_i$),i=1,\ldots,r])
\EndFunction
\end{algorithmic}\label{lbl:partial:order:sorting}
}
\end{algorithm}

\subsection{Example of usage of linear sorting algorithm}

Let us consider family of Bell maximally entangled states for qubits and qudits. In the case of qubits, these states have the following form
\begin{equation}
\mket{\psi}=\alpha \mket{00} \pm \beta \mket{11} \;\;\; \mathrm{or} \;\;\; \mket{\psi}=\alpha \mket{01} \pm \beta \mket{10} \;\;\; \mathrm{and} \;\;\;  |\alpha|^2 + |\beta|^2 = 1,
\end{equation}
and for qubits there exist exactly four such maximally entangled states. The one of the Bell states for qudits where $d=3$ can be written similarly to the qubit case
\begin{equation}
\mket{\psi}=\alpha_0 \mket{00} + \alpha_1 \mket{11} + \alpha_2 \mket{22} \;\;\; \mathrm{and} \;\;\;  |\alpha_0|^2 + |\alpha_1|^2 + |\alpha_2|^2= 1 .
\end{equation}
In general, the set of d-level Bell maximally entangled states for two qudits can be expressed through the following equation:
\begin{equation} 
\mket{\psi^d_{pq}} = \frac{1}{\sqrt{d}} \sum_{j=0}^{d-1} \fun{e}^{2\pi\imath j p / d} \mket{j}\mket{ (j + q) \; \fun{mod} \; d} . \label{lbl:dLevel:Bell:State}
\end{equation}
It is possible to express equation (\ref{lbl:dLevel:Bell:State}) in terms of qudit gates:
\begin{equation}
\mket{\psi^d_{pq}} = {(I_d \otimes X_d)}^q \cdot {(H_d \otimes I_d)} \cdot {(Z_d \otimes I_d)}^p \cdot \fun{CNOT}_{d} \cdot \mket{00} .
\end{equation}
where $0 \leq p \leq d-1$ and $0 \leq q \leq d-1$ are indices of one of $d^2$ allowed Bell state. The symbol $I$ represents the identity matrix for d-level qudit, and $H$ represents the Hadamard gate and $Z$ and $X$ are generalized Pauli's operators.

A simple function written in the Python programming language which uses the QCS module to generate entangled states is depicted in Fig.~\ref{lbl:fig:python:qcs:bell:state:generator}. We use this function to construct entangled states for earlier prepared quantum register.
\begin{figure}
\begin{center}
\begin{tabular}{|l||l|}
\hline
def make\_psi(r, p, q): & def make\_psi(r, p, q): \\
~~~r.Reset() 				& ~~~r.Reset() \\
~~~for i in range(0,q): & ~~~for i in range(0,q): \\
~~~~~~r.NotN(1)			& ~~~~~~r.NotN(1)\\
~~~r.HadN(0) 				& ~~~r.RandGateRealN(0) \\
~~~for i in range(0,p): & ~~~for i in range(0,p):\\
~~~~~~r.PauliZ(0)			& ~~~~~~r.PauliZ(0)\\
~~~r.CNot(0,1)				& ~~~r.CNot(0,1) \\ \hline
\end{tabular}
\end{center}
\caption{The functions written in Python preparing the entangled Bell states for given register. The symbol denoted by $r$ is an object representing the quantum register and $p$ and $q$ are indices of Bell state generated by this procedure. The left column generates maximally entangled state but in the right column instead of the Hadamard gate we use a randomly generated gate to produce a non-maximally entangled states}
\label{lbl:fig:python:qcs:bell:state:generator}
\end{figure}
The function presented in the left column of Fig.~\ref{lbl:fig:python:qcs:bell:state:generator} generates states which have always the same amount of entanglement. Therefore the function from Fig.~\ref{lbl:fig:python:qcs:bell:state:generator} must be equipped with some additional unitary gate to modify the entanglement amount. In the qubit cases the additional rotation gate after Hadamard gate can be used. In general any unitary gate that realises the rotation through any axis may be used to generate Bell states with uniform distribution of entanglement. Indeed, the right version of function \textit{make\_psi} from Fig.~\ref{lbl:fig:python:qcs:bell:state:generator} possesses this feature. 

Using function from Fig.~\ref{lbl:fig:python:qcs:bell:state:generator} and the appropriate computational procedure to calculate the von Neumann entropy it is possible to prepare a simple benchmark. Additionally, to obtain comparable results we prepared simple test as a script in Python language for quantum registers built only from qubits. The test contains the following computation steps: first we generate n quantum registers then for every register the von Neumann entropy is calculated. After these steps we sort the obtained list using the classical method called sorting by selection. In Fig.~\ref{lbl:fig:linear:sort:time} we present the real time necessary to perform this simple test.

\begin{figure}
\begin{center}
\scriptsize
\begin{tabular}{|c||c|}
\hline
Number of registers &  Time (results in secs) \\ \hline
10 & 0.0008762 \\
100 & 0.0048304 \\
1000 & 0.1407390 \\
2000 & 0.4907356 \\
4000 & 1.8053686 \\
10000 & 10.643833  \\
\hline
\end{tabular}
\end{center}
\caption{The time consumed by sorting tests which use randomly generated quantum registers with a different amounts of entanglement in the sense of von Neumann's entropy}
\label{lbl:fig:linear:sort:time}
\end{figure}

\subsection{Computational complexity analysis}

The computational complexity of Algorithm (\ref{lbl:linear:sorting:entropy}) is given by following equation:
\begin{equation}
T(n) = \sum_{i=1}^N \left( T_{1}(n_i,d_i) + T_{2}(n_i,d_i) + T_{3}(n_i,d_i) \right) + T_{sort}(n)
\end{equation}
where $N$ is the total number of quantum registers and $d_i$ represents freedom level of qudit used in given $n_i$ quantum register. Additionally, $T_1(\cdot)$ represents the complexity of partial trace calculation, $T_2(\cdot)$ is the complexity of calculation of the eigenvalues and eigenvectors and $T_3(\cdot)$ is the complexity of the von Neumann entropy calculation. Each of mentioned complexity functions work on matrices and if we assume that $n$ is the size of matrix and $d$ is freedom level of a state which is given by density matrices we obtain the following relations:
\begin{equation}
T_1(n,d) = d n^{2}, \;\;\; T_2(n,d) = n^3, \;\;\; T_3(n,d) = n .
\end{equation}
The complexity of $T_{sort}(n)$ depends on the algorithm used to sort the obtained quantum registers, the value of entropy is used to compare two registers. If we use one of the popular sorting methods like Heapsort with complexity given by $O( n \log(n) )$, the complexity of Algorithm (\ref{lbl:linear:sorting:entropy}) will be
\begin{equation}
T(n) = N( d n^{2} + n^3 + n) + n \log(n) ,
\end{equation}
where the process of computation of eigenvalues and eigenvectors is the most time-consuming part of the whole process of sorting quantum registers.

The second algorithm of sorting quantum states (Algorithm (\ref{lbl:partial:order:sorting})) contains oracle routine as described by Algorithm (\ref{lbl:oracle:two:vec:locc}). The complexity of oracle for the worst case when the ranks are equal is given by:
\begin{equation}
T(n) = n^3 + n \log(n) + n^2 = O(n^3),
\end{equation}
The procedure of calculation the singular value decomposition dominates the computational complexity of oracle routine. It is important to stress that in the oracle routine we also sort the Schmidt coefficients, but by using the classically effective algorithm. However, the SVD still dominates the complexity of \textsc{SD-QueryOracle}. 

It is known \cite{6} that Algorithm (\ref{lbl:partial:order:sorting}) calls the query at most $\mathcal{O}(w \cdot n \log n)$, where $w$ is the maximal width of poset containing $n$ elements but the time of SVD again dominates the whole process of partial sorting of quantum states.

\section{Conclusions and further work}

Basing on the canonical Schmidt decomposition of quantum states a specific semiorder relation has been introduced in the space of quantum states of a given bipartite system. In the case of pure states the introduced semiorder relation possesses a very clear operational meaning as described by Nielsen \cite{3} for the first time. Whether the same operational meaning can be affiliated with the analogous semiorder relation defined in the space of all quantum states should be explained.

Additionally, some version of sorting algorithm of the arising posets, the so called ChainMerge sorting and basing on the particular version of query oracle comparing the amount of entanglement in two quantum states is presented and tested in the case of vector states. The following extensions of the present material seems to be worthwhile to perform: (a) to extend the Nielsen result \cite{3} to cover the case of general quantum states, (b) to formulate several different version of sorting posets algorithm with special emphasis putted on their computational complexity,
(c) to formulate different version of query oracles for comparing the amount of entanglement included in two general states of bipartite systems.



\begin{thebibliography}{99}

\bibitem{Chuang_QCQI}  M.A.~Nielsen, I.L.~Chuang, Quantum Computation and Quantum Information, Cambridge University Press, Cambridge, 2000.

\bibitem{1} M.R.~Garey, D.S.~Johnson, Computers and Intractability: A guide to the Theory of NP-Completeness, W.H.Freeman and Co., San Francisco, 1979.

\bibitem{2} Christos~H.~Papadimitrio, Computational Complexity, Addison Wesley, 1994.

\bibitem{3} M.A.~Nielsen, Conditions for a class of entanglement transformations, Phys. Rev. Lett. 83, 436, 1999.

\bibitem{4} L.~Gurvits, Classical deterministic complexity of Edmonds' problem and Quantum Entanglement, Proc. of the 35th Annual ACM Symposium on Theory of Computing, June 9-11, 2003, San Diego, CA, USA, also available at quant-ph/0303055.

\bibitem{5} C.~Daskalakis, R.M.~Karp, E.~Mossel, S.~Riesenfeld, E.~Verbin, Sorting and Selection In Posets, available at http://arxiv.org/abs/0707.1532.

\bibitem{6} U.~Faigle, Gy.~Turan, Sorting and Recognition Problems for ordered Sets, SIAM J. Comput. 17(1), pp.: 100-113, 1988.

\bibitem{7} R.~Gielerak, M.~Sawerwain, Sorting of amount of Entanglement in Quantum States functions implemented for Quantum Computing Simulator, submitted to KNWS'09 conference (www.knws.uz.zgora.pl).

\bibitem{Terhall} B.M.~Terhall, P.~Horodecki, arXiv:quant-ph/9911117v4.

\bibitem{Chen} K.~Chen, L.A.~Wu, Quant. Inf. Comp., 3, 193, 2003.

\bibitem{Rudolph} O.Rudolph, Phys. Rev. A, 67, 032312, 2003.

\bibitem{q12} R.Gielerak, M.Sawerwain, in preparations.

\bibitem{MSawerwain2008} M.~Sawerwain, R.~Gielerak, Natural quantum operational semantics with predicates, Int. J. Appl. Math. Comput. Sci., 2008, Vol. 18, No. 3, pp.~341-–359.

\end{thebibliography}
\end{document}